\documentclass[10pt]{amsart}

\usepackage{amsfonts}
\usepackage{amssymb}
\usepackage{amsmath,amsthm}
\usepackage{amscd}
\usepackage{graphics}
\usepackage{pdfsync}
\usepackage{mathrsfs}
\usepackage{enumitem}
\usepackage{stmaryrd}
\usepackage[colorlinks,linkcolor=blue,citecolor=blue,urlcolor=blue]{hyperref}

\newcommand{\beq}{\begin{equation}}
\newcommand{\eeq}{\end{equation}}
\newcommand{\bea}{\begin{eqnarray}}
\newcommand{\eea}{\end{eqnarray}}
\newcommand{\bc}{\begin{cases}}
\newcommand{\ec}{\end{cases}}
\newcommand{\nn}{\nonumber}
\newcommand{\noi}{\noindent}

\newtheorem{definition}{Definition}
\newtheorem{proposition}{Proposition}
\newtheorem{theorem}{Theorem}

\newtheorem{lemma}{Lemma}
\theoremstyle{definition}

\newtheorem{remark}{\textbf{Remark}}

\begin{document}

\title[Multivariate Group Entropies]{Multivariate Group Entropies, \\ Super-exponentially Growing Complex Systems \\ and Functional Equations}

\author{Piergiulio Tempesta}
\address{Instituto de Ciencias Matem\'aticas, C/ Nicol\'as Cabrera, No 13--15, 28049 Madrid, Spain\\ and Departamento de F\'{\i}sica Te\'{o}rica, Facultad de Ciencias F\'{\i}sicas, Universidad
Complutense de Madrid, 28040 -- Madrid, Spain}


\email{piergiulio.tempesta@icmat.es, ptempest@ucm.es}
\date{November 18, 2020}
\maketitle

\begin{abstract}
We define the class of multivariate group entropies as a novel set of information - theoretical measures, which extends significantly the family of group entropies. We propose new examples related to the ``super-exponential'' universality class of complex systems; in particular, we introduce a general entropy, representing a suitable information measure for this class. We also show that the group-theoretical structure  associated with our multivariate entropies can be used to define a large family of exactly solvable discrete dynamical models. The natural mathematical framework allowing us to formulate this correspondence is offered by the theory of formal groups and rings.
\end{abstract}


\tableofcontents



\maketitle

\section{Introduction}

The study of generalized entropies in information theory and, in particular, in information geometry, has been actively pursued in the last few decades (see, e.g. \cite{Amari2016book}, \cite{AN2000book} \cite{Khinchin}, \cite{Shannon}, \cite{Shannon2}). 

The combination of group theory with the concept of generalized entropy has led to the notion of \textit{group entropy},  first introduced in \cite{PT2011PRE} and further investigated in  \cite{ST2016PRE}, \cite{PT2016AOP}, \cite{PT2016PRA}, \cite{ET2017JSTAT}, \cite{JT2018ENT}, \cite{RRT2019PRA}, \cite{TJ2020SCIREP}, \cite{CMT2019prep}, \cite{GBP2018prep}. Basically,  a group entropy is a function $S$ 
defined on a space of probability distributions, which possesses several interesting mathematical properties: it satisfies the first three Shannon-Khinchin (SK) axioms (continuity, maximum entropy principle, expansibility) and the \textit{composability axiom}. This axiom requires that given two statistically independent systems, $A$ and $B$, the entropy of the compound system $A\cup B$  must be described by the relation $S(A\cup B)= \Phi(S(A), S(B))$, where $\Phi(x,y)$ is a commutative (formal) group law (see \cite{JT2018ENT} for a recent review).

In \cite{JT2018ENT} and \cite{TJ2020SCIREP}, universality classes of complex systems have been studied in the context of information theory. These classes are defined in terms of a specific state space growth function $\mathcal{W}=\mathcal{W}(N)$, describing the growth of the number of micro-states available as a function of  the number $N$ of identical particles (or constituents) of a given system of the class. Under mild hypotheses, for each class, there exists a specific group entropy, which is extensive over the class and constructed in a purely axiomatic way from the Shannon-Khinchin axioms. In many respects, group entropies represent a versatile tool for generating a wide set of information measures, potentially useful in the theory of complex systems, biology, social sciences, etc.

The main purpose of this note is to further extend this ideas by formalizing and analyzing a new class of group entropies that we shall call \textit{multivariate group entropies}, since they are functions of several variables depending on a probability distribution. This new class has been first considered in \cite{RRT2019PRA}, where  some concrete examples of it have been proposed. We wish to clarify the properties of the multivariate family; in particular, we define a realization in terms of a simple set of functionals: the \textit{multivariate $Z$-entropies}.  They represent a natural generalization of the notion of $Z$-entropy introduced in \cite{PT2016PRA} (which refers to the standard case of \textit{univariate} group entropies).

The theory of group entropies is primarily a \textit{mathematical} one: we are essentially constructing a family of functions on a probability space, which share with R\'enyi's entropy many useful properties.  This allows us to give a natural interpretation of our entropies in the context of information geometry (as in \cite{RRT2019PRA}).


Besides, in this article, we shall construct explicitly new examples of both univariate and multivariate group entropies. They are designed in such a way that they are extensive over the universality class of systems whose state space possesses a \textit{super-exponential growth} in the number $N$ of its constituents.
The \textit{super-exponential class} was first explored in \cite{JPPT2018JPA} within the framework of group entropies. In particular, a new entropy was proposed, extensive over the class characterized by a state space growth rate $\mathcal{W}(N)\sim N^{N}$. We remind the reader that the case of a \textit{sub-exponential} growth rate of the kind $\mathcal{W}(N)\sim N^{a}$ has been largely explored in the context of non-extensive statistical mechanics  \cite{Tsallis1}, \cite{Tsallisbook2009}. A different approach for super-exponential systems, based on scaling-law analysis, has been developed in \cite{KHT2018}.

More precisely, we shall introduce a new family of group entropies parameterized by an ``interpolation function'' that can be easily fine-tuned to cover different universality classes, ranging from the exponential class to the super-exponential one (including the cases of both sub-factorial and super-factorial growth of the state space).

Finally, we establish a connection between the theory of group entropies and that of functional equations. A group entropy inherits in an obvious manner a functional equation, expressed by its own group law. We further investigate the set of functional equations related to group entropies by means of the algebraic structure of \textit{formal rings}, recently introduced in \cite{CT2019prep}. A formal ring is essentially a formal group endowed with a second composition law, which can be naturally realized as a deformed product, compatible with the generalized sum expressed by the group law.

Thus, starting from the formal ring structure associated with a group entropy, we define a set of functional equations,  whose solutions can be explicitly determined.

In this conceptual framework, we propose an application of the present approach to discrete dynamical systems. Indeed, by discretizing over a regular lattice the functional equations related with the formal ring structure of a group entropy, we obtain classes of discrete dynamical systems that, once again, possess exact solutions (here, we do not focus on the initial value problem) and in some cases related sequences of integer numbers.

Although the purpose of this work is essentially of a mathematical nature, we mention that from an applicative point of view, we aim to construct new information measures, potentially useful in the detection of complexity in different scenarios. Indeed, new applications can be envisaged, ranging from statistical inference theory to the study of biological models. An interesting perspective, for instance, is the study of complexity of the human brain. In this regard, one of the most relevant approaches is Tononi's integrated information \cite{TES1998TCS}, \cite{TJ2020SCIREP}, \cite{TSE1994PNAS} which is formulated mathematically in terms of sums of suitable conditional entropies. Also, in \cite{CH2014FHN}, the notion of ``entropic brain'' has been introduced and an entropic approach to consciousness has been developed (for a review, see  \cite{CH2018} and references therein). A different point of view has been suggested in \cite{TJ2020SCIREP}, with the idea of measuring certain features of the brain activity by using  group-theoretical entropic measures different from the standard Shannon information measure.
Another recent perspective, potentially relevant for the theory of complex systems, has been proposed in \cite{ADT2020}. The notion of \textit{permutation entropy}, crucial in data analysis, has been generalized by means of a specific group entropy which stays finite for random processes, when the number of allowed permutations grows super-exponentially with their length. This growth is typical of discrete-time dynamical processes with observational or dynamical noise.
Further work along these lines is in progress. 

The paper is organized as follows. In section \ref{S2}, the axiomatic formulation of the notion of group entropies is presented. In section  \ref{S3}, the theory of formal groups and rings is briefly reviewed.  In section \ref{S4},  the new class of multivariate $Z$-entropies is formally defined and some new examples are discussed in section \ref{S5}. In particular, in section \ref{S6}, a new, general family of super-exponential entropies is introduced. A study of functional equations and discrete dynamical systems naturally associated with group entropies is proposed in the final section \ref{S7}.


\section{On the axiomatic formulation of group entropies} \label{S2}

In order to introduce the notion of group entropy, we first remind the content of the  Shannon-Khinchin (SK) axioms \cite{Shannon}, \cite{Shannon2}, \cite{Khinchin}.  Let $S(p)$ be a function on a set $\mathcal{P}$ of probability distributions $\{p_i\}_{i=1}^{W}$,  with $W> 1$\footnote{We shall identify $W$ with the integer part of $\mathcal{W}(N)$.}, $p_i \geq 0$, $\sum_i p_i=1$. The first three SK axioms essentially amount to the following properties: \\
\noi (SK1) $S(p)$ is continuous with respect to all variables $p_1,\ldots,p_W$. \\
\noi (SK2) $S(p)$ takes its maximum value over the uniform distribution. \\
\noi (SK3) $S(p)$ is expansible: adding an event of zero probability does not affect the value of $S(p)$.

These axioms represent a minimal set of ``non-negotiable'' requirements that the function $S(p)$ should satisfy necessarily in order to be meaningful, both from a physical and information-theoretical point of view.
The fourth axiom, requiring specifically additivity on conditional distributions, leads to Boltzmann's entropy \cite{Khinchin}. Instead, we replace the additivity axiom by a more general statement.

\subsection{Composability axiom} An entropy is said to be \textit{composable} if there exists a sufficiently regular function $\Phi(x,y)$ such that, given two \textit{statistically independent} systems $A$ and $B$,
\beq \label{eq:comp}
S(A\cup B)=\Phi(S(A), S(B)) \ ,
\eeq
where the two systems are allowed to be defined on \textit{any arbitrary probability distribution} of $\mathcal{P}$.
Hereafter, the Boltzmann's constant is assumed to be equal to 1 (and dimensionless).

The relation \eqref{eq:comp} has been introduced in \cite{Tsallisbook2009}. However, we prefer to further specialize this definition (as in \cite{PT2011PRE}, \cite{ST2016PRE}, \cite{PT2016AOP}, \cite{PT2016PRA}, \cite{ET2017JSTAT}, \cite{JT2018ENT}, \cite{RRT2019PRA}, \cite{TJ2020SCIREP}). Precisely,  we shall also require the following properties, that jointly with eq. (1) define the composability axiom: \\
\noi (C1) Symmetry: $\Phi(x,y)=\Phi(y,x)$. \\
\noi (C2) Associativity: $\Phi(x,\Phi(y,z))=\Phi(\Phi(x,y),z)$. \\
\noi (C3) Null-Composability: $\Phi(x,0)=x$. \\

Observe that, indeed, the requirements $(C1)$-$(C3)$ are crucial ones: they impose the independence of the composition process with respect to the order of $A$ and $B$ and the possibility of composing three independent systems in an arbitrary way; besides, they ensure that, when composing a system with another one having zero entropy, the total entropy remains unchanged. In our opinion, these properties are also fundamental: indeed, no thermodynamic or information-theoretical applications would be easily conceivable without them. From an algebraic point of view, a formal power series $\Phi(x,y)$ satisfying the requirements $(C1)$-$(C3)$ defines a \textit{formal group law}.
\\
Consequently, infinitely many choices allowed by relation \eqref{eq:comp}, like $\Phi(x,y)= \sin(x+y)$ or $\Phi(x,y)=x+y+ x^2 y$, or even the simple one $\Phi(x,y)= xy$  are discarded. In this respect, the theory of formal groups \cite{Boch}, \cite{Haze}, \cite{Serre} offers a natural language in order to formulate the theory of generalized entropies in a consistent way.
\\
We remind that another independent and interesting approach, which defines the \textit{pseudo-additivity class} of entropies, has been formulated in \cite{IS2014PHYSA}.
Besides, a different approach to generalized entropies, based on the theory of \textit{statistical inference}, has been proposed in \cite{JK2019PRL}.
\\
First, we shall revise briefly some of the main notions of the theory.
\subsection{Group entropies}



\begin{definition}\label{def1}
A group entropy is a function $S: \mathcal{P}\to \mathbb{R}^{+} \cup \{0 \}$ which satisfies the  Shannon-Khinchin axioms (SK1)-(SK3) and the composability axiom.
\end{definition}

In order to determine a specific group entropy suitable for a certain class of complex systems, we can require a further, natural property. Let $N$ denote the number of identical particles, or constituents of a system; assuming that the state of the system is described by the uniform distribution, we impose that by increasing $N$  the group entropy considered satisfies the following condition:
\beq \label{ext}
\lim_{N\to\infty} \frac{S(N)}{N}=c
\eeq
where $c$ (which might depend on thermodynamic variables) by default is assumed to be a positive quantity, as for instance, in Ref. \cite{TJ2020SCIREP}.

The latter property amounts to require that in the case of maximal indistinguishability  (i.e., in the most ``disordered'' state), asymptotically $S$ grows  linearly as a function of $N$. Clearly, given a state space growth rate $\mathcal{W}=\mathcal{W}(N)$, this requirement selects specific group entropies among the infinitely many allowed by Definition \ref{def1}.


In the subsequent considerations, we shall distinguish two classes of functions: the univariate group entropies, represented by the $Z$-entropies first discussed in \cite{PT2016PRA}, and the multivariate ones.
In  \cite{RRT2019PRA}, the following family of entropy functions has been introduced:
\beq \label{MGE}
S(p):= F(S_{1}(p), \ldots, S_{n}(p)) \ ,
\eeq
where $\{S_{1}(p), \ldots, S_{n}(p)\}$ are all group entropies and $F:\mathbb{R}^{n}\to \mathbb{R}$ is a suitable function. In particular, the conditions ensuring that the entropy \eqref{MGE} is still a group entropy have been
clarified. Besides, a general technique  has been proposed, allowing us to define a new group entropy as the result of a \textit{multivariate composition process} of other group entropies, all of them sharing the same formal group law $\Phi(x,y)$. This reminds us of the typical procedure of Lie group theory, when it is applied to solve partial differential equations: one can generate new exact solutions of an equation from known ones \cite{Olver}. This technique, in turn, shows the power and the  versatility  of the group-theoretical approach for the study of generalized entropies. The previous ideas naturally lead us to define in an abstract way \textit{multivariate group entropies}, namely, group entropies expressed in terms of suitable multivariate functions. This is the main purpose of this work.


In particular, it will be shown that there exists a simple realization of the multivariate class, which will be defined and analyzed in this article: the \textit{multivariate $Z$-entropies}.

In order to make the presentation self-contained, I will first review some basic aspects of the theory of formal groups and rings that provide us with an elegant algebraic language for the formulation of the theory of generalized entropies.
\section{Formal groups and rings} \label{S3}

The mathematical form of the composability axiom leads us naturally to the theory of formal groups  \cite{Boch}, which has found many interesting applications, ranging from  algebraic topology to the theory of elliptic curves, arithmetic number theory, and combinatorics (see also \cite{PT2007CR}, \cite{PT2010ASN}, \cite{PT2015TRAN} for further applications).  Classical reviews are the works   \cite{BMN}, \cite{Nov},   \cite{Serre} and the monograph \cite{Haze} that we shall follow closely for notation.
\par
We will start with an abstract algebraic setting. Let $R$ be a commutative ring  with identity element, and $R\llbracket x_{1},\text{ }%
x_{2},..\rrbracket $ be the ring of formal power series in the variables $x_{1}$, $x_{2}$,
... with coefficients in $R$.

\begin{definition} \cite{Boch}
A commutative one--dimensional formal group law
over $R$ is a formal power series $\Phi \left( x,y\right) \in R\llbracket
x,y\rrbracket $ such that
\begin{equation*}
1)\qquad \Phi \left( x,0\right) =\Phi \left( 0,x\right) =x,
\end{equation*}%
\begin{equation*}
2)\qquad \Phi \left( \Phi \left( x,y\right) ,z\right) =\Phi \left( x,\Phi
\left( y,z\right) \right) \text{.}
\end{equation*}
When $\Phi \left( x,y\right) =\Phi \left( y,x\right) $, the formal group law is
said to be commutative.
\end{definition}
The existence of an inverse formal series $\varphi
\left( x\right) $ $\in R\llbracket x\rrbracket $ such that $\Phi \left( x,\varphi
\left( x\right) \right) =0$ follows from the previous definition.
%
Let $G\left( t\right)$ be the formal series
\beq
G\left( t\right) =\sum_{k=0}^{\infty} a_k \frac{t^{k+1}}{k+1} \label{I.2}.
\eeq
Each coefficient $b_j$ ($j\in \mathbb{N}$) of the inverse series $G^{-1}(s)$ can be explicitly computed. We have  $a_{0}=1, a_{1}=-b_1, a_2= \frac{3}{2} b_1^2 -b_2,\ldots$.
Thus, we can define a formal group law via the formal power series \cite{Haze}
\[
\Phi \left( s_{1},s_{2}\right) =G\left( G^{-1}\left(
s_{1}\right) +G^{-1}\left( s_{2}\right) \right).
\]



Besides, as was proved by Lazard, for any commutative one-dimensional formal group law $\Psi(x,y)$ over a torsionless ring $R$, there exists a unique power series $\psi(x)$ with coefficients in $ R\otimes \mathbb{Q}$,  such that
\[
\psi(x)= x+ O(x^2) \quad
\]\text{and}
\[\quad \Psi(x,y)= \psi\left(\psi^{-1}(x)+\psi^{-1}(y)\right)\in (R\otimes \mathbb{Q}) \llbracket x,y \rrbracket \ .
\]
Standard examples are the additive group law
\beq
\Phi(x,y)=x+y \label{additive}
\eeq
and the multiplicative group law
\beq
\Phi(x,y)= x+y+\sigma x y \ . \label{multiplicative}
\eeq
Let us also review the notion of formal ring, recently introduced in~\cite{CT2019prep}.
\begin{definition}\label{def:formalring}
A formal ring is a triple $\mathcal{A}=(R,\Phi,\Psi)$ where
$\Phi,\Psi\in R\llbracket x,y\rrbracket$ are formal power series such that
\begin{enumerate}
\item $\Phi$ is a commutative formal group law,
\item $\Psi$ satisfies the relations  \begin{align*}
  \Psi(\Psi(x,y),z)&=\Psi(x,\Psi(y,z)) \ ,\\
   \Psi(x,\Phi(y,z))&=\Phi(\Psi(x,y),\Psi(x,z)) \ ,\\
    \Psi(\Phi(x,y),z)&=\Phi(\Psi(x,z),\Psi(y,z)) \ .
 \end{align*}
\end{enumerate}
The formal ring will be said to be commutative if
$\Psi(x,y)=\Psi(y,x)$.
\end{definition}
As has been proved in \cite{CT2019prep}, a one-dimensional realization of the previous definition is given by
\begin{eqnarray} \label{eq:FRE1}
\Phi \left( s_{1},s_{2}\right) &=& G\left( G^{-1}\left(
s_{1}\right) +G^{-1}\left( s_{2}\right) \right)  \\ \label{eq:FRE2}
\Psi \left( s_{1},s_{2}\right) &=& G\left( G^{-1}\left(
s_{1}\right) \cdot G^{-1}\left( s_{2}\right) \right)
\end{eqnarray}
for a suitable $G(t)\in R\llbracket t\rrbracket$. Besides, an $n$-dimensional generalization of these constructions has also been proposed in \cite{CT2019prep}.
\\
A simple example of a one-dimensional formal ring structure is provided by the couple $\{\Phi(x,y), \Psi(x,y)\}$ where $\Phi(x,y)$ is given by eq. \eqref{multiplicative}  jointly with this new, formal expression for  the ``product'':
\beq \label{Tproduct}
\Psi(x,y)= \frac{\exp\big((1/\sigma)\ln(1+\sigma x)\ln(1+\sigma
    y)\big)-1}\sigma \ .
\eeq
\begin{remark}
In the forthcoming applications to group entropies, we will realize the equations \eqref{eq:FRE1}, \eqref{eq:FRE2} in terms of standard real-valued functions (this, a priori, may require suitable constraints: for instance, $x, y\geq0$ and $\sigma>0$ in eq. \eqref{Tproduct}).
\end{remark}
\section{The multivariate group entropies} \label{S4}
We shall present now the main construction of this article. To this aim, we remind that the notion of \textit{group logarithm} \cite{PT2016PRA} provides us with a group-theoretical deformation of the standard logarithm. It can be defined with different regularity properties.
For the purposes of this paper, it is sufficient to think of a group logarithm as a function $\ln_{\chi}(x):=\chi(\ln x)$, where $\chi\in C^{1}(\mathbb{R}_{\geq 0})$ is a strictly increasing function,
 taking positive values over $\mathbb{R}^{+}$, vanishing at zero. A natural choice is to assume, in addition, the form $\chi(t)=t+O(t^2)$ as $t\to 0$, so that we have $\Phi(x,y)= x+y+ \ldots$.
Let us introduce a set of real positive \textit{entropic parameters} $\{\alpha_1, \ldots, \alpha_n \}$. A further, typical restriction  is to assume that $\ln_{\chi}$ is a strictly concave function and to define
the parameters $\alpha_i$ in the range $(0, 1)$. More generally, if one deals with strictly Schur-concave functionals (which is sufficient for many purposes), this restriction can be dropped.

\begin{definition}
A multivariate $Z$-entropy (MZE) is a function $Z: \mathcal{P}\to \mathbb{R}^{+} \cup \{0 \}$ of the form
\beq \label{MZE}
Z_{\chi,\alpha_1,\ldots,\alpha_n}(p):= 
\ln_{\chi}\bigg(\bigg(\sum_{i_{1}=1}^{W}p_{i_{1}}^{\alpha_{1}}\bigg)^{\frac{1}{1-\alpha_{1}}}\bigg(\sum_{i_{2}=1}^{W}p_{i_{2}}^{\alpha_{2}}\bigg)^{\frac{1}{1-\alpha_{2}}}\cdot \ldots\cdot \bigg(\sum_{i_{n}=1}^{W}p_{i_{n}}^{\alpha_{n}}\bigg)^{\frac{1}{1-\alpha_{n}}}\bigg ) \ ,
\eeq
where  $n\in\mathbb{N}\backslash \{0\}$.
\end{definition}
\noi The univariate $Z$-entropies \cite{PT2016PRA} (for short, $Z$-entropies) have the form
\beq
Z_{G,\alpha}(p)= \frac{\ln_{G}(\sum_{i=1}^{W}p_{i}^{\alpha})}{1-\alpha} \ .
\eeq
\noi 
Thus, for $n=1$, the multivariate class is in a direct relation with the $Z$-class: each entropy of the first class   gives an entropy of the second one with a suitable $G(t)$ obtained by means of the identification
$
\chi(t)= \frac{G\big((1-\alpha)t\big)}{1-\alpha} \ .
$

It is straightforward to prove the following
\begin{proposition}
The MZEs satisfy the first three Shannon-Khinchin axioms.
\end{proposition}

\noi The main result of this section concerns the composability property of the MZE-class.
\begin{theorem}
The multivariate $Z$-entropies \eqref{MZE} satisfy the composability axiom, with composition law given by $\Phi(x,y)=\chi\big(\chi^{-1}(x)+\chi^{-1}(y)\big)$.
\end{theorem}
\begin{proof}
Let $p_{ij}^{A\times B}=p_{i}^{A}p_{j}^{B}$ be the joint probability distribution associated with the statistically independent systems $A$ and $B$. Then, we have that the entropies \eqref{MZE} satisfy the relation
\begin{eqnarray}
\nonumber && Z_{\chi,\alpha_1,\ldots,\alpha_n}(p^{A\times B}) = \nonumber \ln_{\chi}\bigg(\prod_{k=1}^{n}\bigg( \sum_{i_{k},j_{k}}p_{i_{k}j_{k}}^{\alpha_{k}}\bigg)^{\frac{1}{1-\alpha_{k}}} \bigg) \\ \nonumber &=& \ln_{\chi}\bigg( \prod_{k=1}^{n}\bigg(\sum_{i_{k}}p_{i_{k}}^{\alpha_{k}}\bigg)^{\frac{1}{1-\alpha_{k}}}\prod_{k=1}^{n}\bigg(\sum_{j_{k}}p_{j_{k}}^{\alpha_{k}}\bigg)^{\frac{1}{1-\alpha_{k}}}\bigg) \\
\nonumber &=& \chi\bigg(\ln \bigg(\prod_{k=1}^{n}\bigg(\sum_{i_{k}}p_{i_{k}}^{\alpha_{k}}\bigg)^{\frac{1}{1-\alpha_{k}}}\bigg)+ \ln \bigg(\prod_{k=1}^{n}\bigg(\sum_{j_{k}}p_{j_{k}}^{\alpha_{k}}\bigg)^{\frac{1}{1-\alpha_{k}}}\bigg) \bigg)= \\
\nonumber &=& \chi\bigg(\chi^{-1}\bigg(\chi\bigg(\ln \bigg(\prod_{k=1}^{n}\bigg(\sum_{i_{k}}p_{i_{k}}^{\alpha_{k}}\bigg)^{\frac{1}{1-\alpha_{k}}}\bigg)\bigg)\bigg)+\chi^{-1}\bigg( \chi\bigg( \ln \bigg(\prod_{k=1}^{n}\bigg(\sum_{j_{k}}p_{j_{k}}^{\alpha_{k}}\bigg)^{\frac{1}{1-\alpha_{k}}} \bigg)\bigg)  \bigg)\bigg)=
\\ \nonumber &=& \chi\bigg(\chi^{-1}\bigg( Z_{\chi,\alpha_1,\ldots,\alpha_n}(p^{A})\bigg)+ \chi^{-1}\bigg(Z_{\chi,\alpha_1,\ldots,\alpha_n}(p^{B})\bigg)\bigg)=
\\ \nonumber &=& \Phi\big(Z_{\chi,\alpha_1,\ldots,\alpha_n}(p^{A}), Z_{\chi,\alpha_1,\ldots,\alpha_n}(p^{B})\big) \ .
\\
\end{eqnarray}
\end{proof}

According to the previous discussion, and under the same hypotheses, the multivariate $Z$-entropies represent a general class of group entropies.

The physical and information-theoretical meaning of the function $\chi$ can be grasped by following the approach introduced in \cite{JT2018ENT} and \cite{TJ2020SCIREP}. Precisely,
let us assume that the state space growth function $\mathcal{W}\in C^{1}(\mathbb{R}^{+})$ associated with a universality class of complex systems is strictly increasing and takes positive values over $\mathbb{R}^{+}$. Thus, we can construct a related multivariate group entropy satisfying the extensivity requirement, by means of the formula


\beq\label{eq:MZErepr}
Z_{\chi,\alpha_1,\ldots,\alpha_n}(p)=\lambda \left(\mathcal{W}^{-1}\left(\bigg(\sum_{i_{1}=1}^{W}p_{i_{1}}^{\alpha_1}\bigg)^{\frac{1}{1-\alpha_1}}\bigg(\sum_{i_{2}=1}^{W}p_{i_{2}}^{\alpha_2}\bigg)^{\frac{1}{1-\alpha_2}}\cdot\ldots\cdot \bigg(\sum_{i_{n}=1}^{W}p_{i_{n}}^{\alpha_n}\bigg)^{\frac{1}{1-\alpha_n}}\right)- \mathcal{W}^{-1}(1) \right)\ ,
\eeq
where $\lambda\in\mathbb{R}^{+}$, $\alpha_1,\ldots,\alpha_n >0$. If we assume
\beq\label{eq:assump}
(\mathcal{W}^{-1})'(1)\neq 0 \ , 
\eeq
then a natural choice is $\lambda = \frac{1}{(\mathcal{W}^{-1})'(1)}$.
The form \eqref{eq:MZErepr} will be taken into account in the forthcoming considerations.

The associated composition law  is given by
\beq
\phi(x,y)=\lambda \left\{\mathcal{W}^{-1}\Big[
\mathcal{W}\left(\frac{x}{\lambda}+\mathcal{W}^{-1}(1)\right) \mathcal{W}\left(\frac{y}{\lambda}+\mathcal{W}^{-1}(1)\right)
\Big] -\mathcal{W}^{-1}(1)\right\} \ .
\label{phi_W}
\eeq
This formula, for the univariate case only, has been derived in \cite{JT2018ENT} and rigorously proved in \cite{TJ2020SCIREP}. Besides, one can show that it is still valid in the multivariate case.

\section{New examples of multivariate $Z$-entropies}\label{S5}
To our knowledge, the first examples of multivariate functionals have been proposed in \cite{RRT2019PRA}. In this section, we shall present new cases of multivariate entropies. 

\subsection{A linear combination of R\'enyi's entropies}
Let $R_{\alpha}(p)=\frac{1}{1-\alpha}\ln (\sum_{i=1}^{W} p_{i}^{\alpha})$ denote the R\'enyi entropy \cite{Renyi1960}. The simplest multivariate example is of course provided by the linear combination
\beq
R_{\alpha_1,\ldots, \alpha_n}(p)= \lambda_1 R_{\alpha_1}(p)+\ldots +\lambda_n R_{\alpha_n}(p)\ ,
\eeq
where $\lambda_1,\ldots,\lambda_n\in\mathbb{R}^{+}$, $n\in\mathbb{N}\backslash \{0\}$.
The composition law is obviously the additive one.

\subsection{An entropy for a stretched-exponential-type growth rate}
Let us consider the universality class of systems whose state space growth rate is
\beq
\mathcal{W}(N)= \exp\big({\gamma N^{\beta}}\big), \qquad \beta>1, \quad \gamma\in\mathbb{R}^{+}\ .
\eeq
We introduce the entropy
\begin{equation}
Z_{ \gamma,\beta, \alpha_1,\ldots,\alpha_n}(p):=\left(\frac{\ln\big(\sum_{i_1=1}^{W} p_{i_1}^{\alpha_1}\big)}{\gamma(1-\alpha_1)}+\ldots +\frac{\ln\big(\sum_{i_n=1}^{W} p_{i_n}^{\alpha_n}\big)}{\gamma(1-\alpha_n)}\right)^{1/\beta} \ .
\end{equation}
The related composition law is given by
\beq
\Phi(x,y)=\big(x^{\beta}+y^{\beta}\big)^{\frac{1}{\beta}} \ ,
\eeq
which was considered in \cite{PT2016AOP} in the trace-form case \footnote{Note that here we do not require the expansion condition for $\chi(t)$ around zero. Clearly, this condition can be easily restored by properly shifting $\mathcal{W}(N)$, which does not affect its asymptotic behaviour.}. Its univariate version
\beq
Z_{\gamma,\beta,\alpha}(p):=\left(\frac{\ln\big(\sum_{i=1}^{W} p_{i}^{\alpha}\big)}{\gamma(1-\alpha)}\right)^{1/\beta}
\eeq
is by construction a composable entropy. Its functional form reminds the trace-form functional known as the $\delta$-entropy, which was introduced in \cite{TC2013EPJC} (however, it is non-composable for $\delta\neq1$).

The ``stretched'' case  $0<\beta<1$ can be treated in a similar way.
\subsection{A rapidly growing state space}
Let us consider the growth function
\beq \label{eq:supsup}
\mathcal{W}(N) = \exp\big(k_2 \exp(k_3 N)-k_1\big) \ ,
\eeq
with $k_{1},k_{2},k_{3}\in\mathbb{R}^{+}$. A new group entropy, extensive over the class of systems characterized by the growth rate \eqref{eq:supsup}, is
\[
Z_{\alpha,k_1}(p):=k_1\ln \bigg(\frac{\ln\big(\sum_{i}p_{i}^{\alpha}\big)}{k_1 (1-\alpha)}+1 \bigg) \ .
\]
Its multivariate extension is straightforward. According to eq. \eqref{phi_W}, the composition law associated  reads
\beq
\Phi(x,y)= k_1 \ln\big(e^{\frac{x}{k_1}} + e^{\frac{y}{k_1}}-1\big)  \ .
\eeq

\section{A new, general entropy for the super-exponential class}\label{S6}

In this section, we shall introduce an entropy designed to be extensive in different super-exponential regimes.
In this context, we are naturally led to the classical functional equation
\beq
y e^y = x \ .
\eeq
For $x\geq 0$, it admits the real solution $y=L(x)$; we denote here by $L(x)$ the branch  of the real $W$-Lambert function commonly denoted by $W_{0}(x)$. \footnote{This to avoid confusion with the growth function $\mathcal{W}(x)$.}

\noindent For our purposes, let us consider the state space growth rate
\beq\label{eq:super}
\mathcal{W}(N) = e^{g(N \ln N)} \ ,
\eeq
where $g$ is a strictly increasing $C^{1}(\mathbb{R}_{\geq 0})$ invertible function, taking positive values on $\mathbb{R}_{>0}$, with
\beq
g(x)\to\infty \quad \text{for}~x\to\infty.
\eeq
In the spirit of formal group theory, we also assume
\beq \label{eq:g}
g(x)=x+O(x^2) \quad \text{as} \quad x\to 0 \ .
\eeq

For  simplicity of notation, let us denote by $\gamma(x)$ the compositional inverse of $g$: $\gamma(g(x))=x$. We observe that $\gamma'(1)\neq 0$, coherently with the assumption \eqref{eq:assump}.

The function $g(x)$ plays the role of ``interpolating function'', allowing us to deform the growth rate of the considered state space and, consequently, to explore different regimes.

The natural choice $g(x)=\nu x$, with $\nu\in\mathbb{R}^{+}$, corresponds to the interesting case
\beq\label{eq:mu-super}
\mathcal{W}(N)=N^{\nu N} \ ,
\eeq
which was studied in \cite{JPPT2018JPA}. Indeed, a Hamiltonian model, called the \textit{pairing model}, was introduced as an example of a complex system whose state space growth rate is expressed by the function \eqref{eq:mu-super}.
The growth rate \eqref{eq:super} depends on the asymptotic behaviour of $g(x)$ when $x\to\infty$.

\noi We propose the new (univariate) entropy
\beq\label{Zgamma}
Z_{\gamma(x),\alpha}(p):= \exp\bigg[L\bigg(\gamma\bigg(\frac{\ln(\sum_{i=1}^{W}p_{i}^{\alpha})}{1-\alpha}\bigg)\bigg)\bigg] - 1 \ ,
\eeq
which can be extensive in different regimes. In order to illustrate the potential relevance of the functional \eqref{Zgamma}, let us consider the limit
\beq
l=\lim_{x\to\infty} \frac{g(x)}{x} \ .
\eeq
We can distinguish several cases.
\par
\noi a) If $g(x)$ for $x\to\infty$ grows faster than any linear function  (i.e. $l=\infty$), the corresponding entropy \eqref{Zgamma} is extensive in  \textit{super-factorial} regimes. Clearly, these regimes can be further discriminated by specific choices of $g(x)$ (and consequently of $\gamma(x)$).
\\
b) If $g(x)$ is a linear function, with $l=\nu\in\mathbb{R}^{+}$, entropy \eqref{Zgamma}  reproduces the group entropy introduced in \cite{JPPT2018JPA}. In particular,  for $0<\nu<1$ entropy \eqref{Zgamma} is extensive in the sub-factorial regime. \\
c) If
\beq \label{eq:gexp}
g(x) = \exp(L(x))-1, \qquad \mathcal{W}(N)\sim e^{N} \ ,
\eeq
($l=0$), entropy \eqref{Zgamma} is extensive in the standard exponential regime.

\noi Thus, according to the previous discussion, the functional form \eqref{Zgamma} enables us to ``sweep out'' several universality classes in a simple and direct way.


\noi Obviously, for $\gamma(x) = L^{-1}(\ln (x +1))$, we obtain the R\'enyi entropy.

\subsection{The group-theoretical structure} In full generality, the composition law satisfied by the class of entropies \eqref{Zgamma} reads
\beq \label{eq:Phiexpgamma}
\Phi(x,y)= \exp\bigg[L\big(\phi\big((x+1)\ln(x+1),(y+1)\ln(y+1)\big) \big)\bigg]-1 \ ,
\eeq
where
\[
\phi(x,y)=g^{-1}(g(x)+g(y))
\]
is the composition law induced by the function $g(x)$. Therefore, eq. \eqref{eq:Phiexpgamma} represents a family of group-theoretical structures, depending on the choice of $g(x)$
.
We expect that the class of entropies \eqref{Zgamma}, as well as its multivariate version
\beq \label{Zgammamult}
Z_{\gamma(x),\alpha_1,\ldots,\alpha_n}(p):= \exp\bigg[L\bigg(\gamma\bigg(\frac{ \ln(\sum_{i_{1}=1}^{W}p_{i_{1}}^{\alpha_1})}{1-\alpha_1}+\ldots + \frac{\ln(\sum_{i_{n}=1}^{W}p_{i_{n}}^{\alpha_n})}{1-\alpha_n}\bigg)\bigg)\bigg] - 1 \
\eeq
might be useful in applications to the theory of complex systems, data analysis as well as in the study of random processes. These issues will be considered elsewhere.

\section{Formal Rings, Functional Equations and Discrete Systems}\label{S7}
Group entropies by construction are intimately related with the theory of functional equations \cite{Aczel2006book}. Indeed, according to the composability axiom, each group entropy satisfies a specific functional equation, expressing its own composition law. 

In this section, inspired by the notion of formal rings, we shall propose a Lemma that establishes, in the one-dimensional case, a more general connection among entropies and functional equations. Also, we shall introduce 
families of difference equations, which admit exact solutions constructed, once again, by means of the algebraic structure associated with group entropies.

The following simple result is at the basis of this correspondence.

\begin{lemma}\label{lemma}
Let  $\{\Phi(x,y), \Psi(x,y)\}$ be two compatible composition laws of the form
\[
\Phi(x,y)=G(G^{-1}(x)+G^{-1}(y)) \ , \quad  \Psi(x,y)=G(G^{-1}(x)\cdot G^{-1}(y)) \ ,
\]
where $G(t)$  is a real-valued and invertible function, defined on $\mathbb{R}$  (with $G(t)= t+ O(t^2)$ as $t\to 0$). \\
\noi \textit{i)} The functional equation
\beq
f(x+y)=\Phi(f(x),f(y)) \label{master1}
\eeq
admits the solution given by
\[
f(x)= G(x) \ .
\]
ii) The functional equation
\beq
f(x+y)=\Psi(f(x),f(y)) \label{master2}
\eeq
admits the solution given by
\[
f(x)= G(\exp(x)) \ .
\]
\noi \textit{iii)}  Let $x,y>0$. The functional equation
\beq
f(xy)=\Phi(f(x),f(y)) \label{master3}
\eeq
admits the solution 
\[
f(x)=G(\ln (x))  \ .
\]
\noi\textit{iv)} The functional equation
\beq
f(xy)=\Psi(f(x),f(y)) \label{master4}
\eeq
admits the solution
\[
f(x)=G(x) \ .
\]
\end{lemma}

\vspace{3mm}

\begin{proof} We shall treat explicitly the cases \textit{ii)} and \textit{iii)} only. The proof of the other cases is immediate.
\\
\\
\textit{ii)} Equation \eqref{master2} can be written in the form
\[
f(x+y)=\Psi(f(x),f(y))=G(G^{-1}(f(x))\cdot G^{-1}(f(y))) \ .
\]
Therefore, for $f(x):=G(\exp(x))$, we have the identity
\[
f(x+y)= G(G^{-1}(G(\exp(x)))\cdot G^{-1}(G(\exp(y))))=G(\exp(x+y)) \ .
\]
\noi \textit{iii)} We have
\bea
\nn \Phi(f(x),f(y))&=&G\left(G^{-1}(f(x))+G^{-1}(f(y))\right)\\ &=&G\left(G^{-1}(G(\ln x))+G^{-1}(G(\ln y))\right) \\
\nn &=&G(\ln x+\ln y)= f(xy) \ .
\eea
\end{proof}
Hereafter, we shall propose  several examples of application of the previous results to the study of certain families of functional equations; all quantities are assumed to be real. The trivial case is represented by the standard Cauchy functional equation
\beq
f(x+y)=f(x)+f(y) \ ,
\eeq
which corresponds to the additive group law \eqref{additive} and to the choice $G(t)=t$. We have immediately the linear solution $f(x)= a~x$, $a\in\mathbb{R}$.
\subsection{The formal ring structure associated with Tsallis entropy}
The formal ring structure given by eqs. \eqref{multiplicative} and \eqref{Tproduct} allows us to define the functional equations
\beq \label{F1T}
f(x+y)=f(x)+f(y)+\sigma f(x)f(y) \ ,
\eeq
\beq \label{F2T}
f(x\cdot y)=  \frac{\exp\big((1/\sigma)\ln(1+\sigma f(x))\ln(1+\sigma
    f(y))\big)-1}\sigma \ ,
\eeq
with $\sigma\in \mathbb{R}$. The ring structure is generated by the group exponential $G(t)= \frac{\text{exp}(\sigma t)-1 }{\sigma}$, which provides an exact solution to eqs. \eqref{F1T}, \eqref{F2T} according to Lemma \ref{lemma} .

\subsection{A rational group law}
Consider the formal group law
\beq
\Phi(x,y) = \frac{x+y+a xy}{1+b xy} \label{rationlaw} \ .
\eeq
Notice that when $a=b=0$, we obtain the standard additive law \eqref{additive}; for $b=0$, we are led to the multiplicative law \eqref{multiplicative}.  We mention that a one-parametric reduction of the formal group \eqref{rationlaw}, i.e
\beq
\Phi(x,y) = \frac{x+y+(\alpha-1) xy}{1+\alpha xy} \label{philaw}
\eeq
plays an important role in algebraic topology. Precisely, for $\alpha=-1,0,1$, we obtain the group laws associated with the Euler characteristic, the Todd genus and the Hirzebruch $L$-genus \cite{BMN}, respectively. The functional equation
\beq
f(x+y)= \frac{f(x) +f(y)+ a f(x) f(y)}{1+b f(x)f(y)}
\eeq
with $a,b>0$ admits the solution
\beq
f(x)=\frac{2 (e^{r x}-1)}{-a(e^{r x}-1)+\sqrt{a^2+4b}(e^{r x}+1) }, \quad x, r\in \mathbb{R} \label{Grational} \ .
\eeq
This solution coincides with the ``group logarithm'' $G(x)$ introduced in \cite{CTT2016AOP}, where a generalized, bi-parametric Tsallis entropy has been proposed.
In turn, the multiplicative equation
\[
f(xy)= \frac{f(x) +f(y)+ a f(x)f(y)}{1+b f(x)f(y)}
\]
admits the solution
\[
f(x)= \frac{2 (x^{r}-1)}{-a(x^{r}-1) + \sqrt{a^2+4b}\hspace{1mm}(x^{r}+1) }, \qquad x>0,\quad r\in \mathbb{R} \ . \label{genlogplus}
\]
The formal product $\Psi(x,y)$ associated with the group law \eqref{philaw} has been computed in \cite{CT2019prep}. Thus, a similar analysis can be extended to the other functional equations related with this ring structure.

In \cite{BK1991SBO}, interesting examples of functional equations arising from  generalized cohomology theories have been discussed  in the context of formal group theory.
\subsection{Discrete Systems and Sequences of Integer Numbers}
A direct connection between group entropies and discrete dynamical systems can also be established. To this aim, we shall consider a regular lattice of points parametrized by $n\in\mathbb{Z}$.

Given a formal ring structure $\mathcal{A}=\{\mathbb{R}, \Phi(x,y), \Psi(x,y)\}$, we can introduce the set of discrete equations
\begin{eqnarray}
z_{n+m} &=& \Phi\left(z_n, z_m\right) \hspace{30mm} \text{(DE1)} \\ \label{diffeq1}
z_{n+m}&=&\Psi\left(z_n, z_m\right)\hspace{30mm} \text{(DE2)} \\ \label{diffeq2}
z_{n\cdot m}&=& \Phi\left(z_n, z_m\right)  \hspace{30mm}  \text{(DE3)} \label{diffeq3} \\
z_{n\cdot m}&=&\Psi\left(z_n, z_m\right)  \hspace{30mm}  \text{(DE4)}\label{diffeq4}
\end{eqnarray}
obtained from the functional equations \eqref{master1}-\eqref{master4} by means of the correspondence $z_n:= f(n), \hspace{1mm} n\in\mathbb{Z} \ .$ 
Exact solutions of eqs. (DE1)-(DE4) can be constructed by applying Lemma \ref{lemma}.
As a specific example, let us consider the realization of (DE1) given by the equation
\[
z_{n+m}=z_n+z_m+ p~z_n z_m, \qquad n,m \in\mathbb{Z}, \quad p\in\mathbb{R} \ .
\]
The group-theoretical solution previously obtained reads now
\[
z_n=\frac{\exp (p n)-1}{p} \ .
\]
The related form \eqref{diffeq3}
\beq \label{eq:Tsallisrec}
w_{n\cdot m}=w_n+w_m+ p~w_n w_m , \qquad p\in\mathbb{R} \ ,
\eeq
admits the solution $w_n= G\big(\ln (z_n)\big)= \frac{n^{p}-1}{p}$.
\\
\noi We mention that for $p\in\mathbb{N}\backslash\{0\}$, we can easily generate sequences of integers of some interest. Indeed, by means of a re-scaling, we obtain that the function
$q_n:= {n^{p}-1}$ satisfies the recurrence
\beq \label{eq:qseq}
q_{n\cdot m}=q_n+q_m+ q_n q_m, \qquad n, m\in\mathbb{Z} \ .
\eeq
For instance, for $p=3$, we generate the sequence of integers
\begin{eqnarray*}
&& \ldots, -513, -344, -217, -126, -65, -28, -9, -2, -1, 0,\\ && 7, 26,  63, 124, 215, 342, 511, \ldots
\end{eqnarray*}
or, for $p=5$, the sequence of integers
\begin{eqnarray*}
&\dots,& -32769, -16808, -7777, -3126, -1025, -244, -33,\\ && -2, -1, 0, 31, 242,
\nn 1023, 3124,~7775,~16806,~32767, \ldots
\end{eqnarray*}
all of them satisfying recurrence \eqref{eq:qseq}.

\section*{Data Availability Statement}
The data that supports the findings of this study are available within the article.

\section*{Acknowledgement}
P. T. wishes to thank J. M. Amig\'o, M. A. Rodr\'iguez and G. Sicuro for useful discussions. This work has been partly supported by the research project
PGC2018-094898-B-I00, Ministerio de Ciencia, Innovaci\'on y Universidades, Spain, and by the ICMAT Severo Ochoa project
SEV-2015-0554, Ministerio de Ciencia, Innovaci\'on y Universidades. P.T. is member of the Gruppo Nazionale di Fisica Matematica (INDAM), Italy.

\end{document}